\documentclass[reqno,12pt,a4paper]{amsart}

\usepackage{amssymb}
\usepackage{url}
\usepackage{hyperref}

\allowdisplaybreaks[4]

\usepackage{mathrsfs}
\let\mathcal\mathscr
\usepackage[mathscr]{eucal}

\newcommand*{\pd}[2]{\mathchoice{\frac{\partial#1}{\partial#2}}
  {\partial#1/\partial#2}{\partial#1/\partial#2}
  {\partial#1/\partial#2}}

\let\phi=\varphi
\let\kappa=\varkappa

\DeclareMathOperator{\sym}{sym}

\newcommand*{\sd}[2]{\{\,#1\mid#2\,\}}
\newcommand*{\eval}[1]{\left.#1\right|}
\newcommand*{\abs}[1]{\left|#1\right|}
\newcommand*{\Ev}{\mathbf{E}}

\theoremstyle{theorem}
\newtheorem{proposition}{Proposition}
\numberwithin{proposition}{section}
\newtheorem{corollary}{Corollary}
\numberwithin{corollary}{section}
\newtheorem{theorem}{Theorem}
\numberwithin{theorem}{section}
\newtheorem{lemma}{Lemma}
\numberwithin{lemma}{section}
\theoremstyle{definition}

\theoremstyle{remark}

\usepackage[all]{xy}
\SelectTips{cm}{}

\usepackage{mathrsfs}
\let\mathcal\mathscr
\usepackage[mathscr]{eucal}

\newcommand{\cprime}{\/{\mathsurround=0pt$'$}}

\begin{document}
\date{\today} \title[Reductions and their symmetry properties]{Reductions of
  the universal hierarchy and rdDym equations and their symmetry properties}

\author{P.~Holba} \address{Mathematical
  Institute, Silesian University in Opava, Na Rybn\'{\i}\v{c}ku 1, 746 01
  Opava, Czech Republic} \email{M160016@math.slu.cz}
\author{I.S.~Krasil{\cprime}shchik} \address{Independent University of Moscow,
  B. Vlasevsky 11, 119002 Moscow, Russia
  \& Trapeznikov
  Institute of Control Sciences, 65 Profsoyuznaya street, Moscow 117997,
  Russia} \email{josephkra@gmail.com} \author{O.I.~Morozov} \address{Faculty
  of Applied Mathematics, AGH University of Science and Technology,
  Al. Mickiewicza 30, Krak\'ow 30-059, Poland}
\email{morozov{\symbol{64}}agh.edu.pl} \author{P.~Voj{\v{c}}{\'{a}}k}
\address{Mathematical Institute, Silesian University in Opava, Na
  Rybn\'{\i}\v{c}ku 1, 746 01 Opava, Czech Republic}
\email{Petr.Vojcak@math.slu.cz}

\date{\today}

\begin{abstract}
  We consider the equations $u_{yy} = u_y u_{xx} - (u_x + u)u_{xy} + u_x u_y$
  and $u_{yy} = (u_x + x)u_{xy} - (u_{xx} + 2)u_y$ that arise as reductions of
  the universal hierarchy and rdDym equations, respectively, and describe the
  Lie algebras of nonlocal symmetries in infinite-dimensional coverings
  naturally associated to these equations.
\end{abstract}



\maketitle

\tableofcontents
\newpage

\section*{Introduction}

In a series of recent
papers~\cite{B-K-M-V-2014,B-K-M-V-2015,B-K-M-V-2018,H-K-M-V-jnmp} we studied
symmetry and integrability properties of the four linearly degenerate 3D
equations,~\cite{Fer-Moss-2015}. In particular, in~\cite{B-K-M-V-2015} we
described 2D reductions of the Pavlov, universal hierarchy and rdDym equations
that possess differential coverings of the rational form. In the recent
paper~\cite{H-K-M-V-jnmp}, using the reduction techniques, we showed that for
the two of these reductions (one of them being the Gibbons-Tsarev equation)
the Lie algebra of nonlocal symmetries is isomorphic to the Witt algebra
\begin{equation*}
  \mathfrak{w} = \sd{z^{i+1}\pd{}{z}}{i\in\mathbb{Z}}
\end{equation*}
of polynomial vector fields.

In the current paper, we prove similar results for the reduction
\begin{equation*}
  u_{yy} = u_y u_{xx} - (u_x + u)u_{xy} + u_x u_y 
\end{equation*}
of the universal hierarchy equation (Section~\ref{sec:univ-hier-equat}) and
for the equation
\begin{equation*}
  u_{yy} = (u_x + x)u_{xy} - (u_{xx} + 2)u_y,
\end{equation*}
which is a reduction of the rdDym equation
(Section~\ref{sec:rddym-equation}). Namely, we show that in the first case the
algebra of nonlocal symmetries is isomorphic
to~$\mathfrak{w} \oplus \mathfrak{s}_2$, where~$\mathfrak{s}_2$ is the
two-dimensional solvable Lie algebra
(Theorem~\ref{sec:thm-nonl-symm-reduct-uhe}), while in the second case this
algebra is~$\mathfrak{w} \oplus \mathfrak{a}_1$, where~$\mathfrak{a}_1$ is the
one-dimensional Abelian Lie algebra, see
Theorem~\ref{sec:thm-nonl-symm-reduct-rd}.

In Section~\ref{sec:diff-cover-nonl}, we introduce the necessary definitions
and constructions. Section~\ref{sec:discussion} contains a short discussion of
the obtained results.

\section{Differential coverings and nonlocal symmetries}
\label{sec:diff-cover-nonl}

Here we briefly discuss the necessary facts from nonlocal geometry of
PDEs. See details in~\cite{AMS-book,Kras-Vin-Trends-1989}.

Let~$\mathcal{E}\subset J^\infty(n,m)$ be an infinitely prolonged differential
equation (or a system of equations) in unknowns $u^j(x^1,\dots,x^n)$,
$j=1,\dots,m$, embedded to the corresponding infinite jet space. Denote
by~$u_\sigma^j$ jet coordinates and assume that~$\mathcal{E}$ is defined by a
system of relations $F^\alpha(x^1,\dots,x^n,\dots,u_\sigma^j,\dots) = 0$,
$j=1,\dots,l$. Denote by
\begin{equation*}
  D_i = \pd{}{x^i} + \sum u_{\sigma i}^j\pd{}{u_\sigma^j} 
\end{equation*}
the total derivatives on~$\mathcal{E}$. Let
\begin{equation*}
  \ell_{\mathcal{E}} =
  \begin{pmatrix}
    \sum_\sigma \pd{F^\alpha}{u_\sigma^j}D_\sigma
  \end{pmatrix}
\end{equation*}
be the linearization of~$\mathcal{E}$, where~$D_\sigma$ is the composition of
the total derivatives corresponding to the multi-index~$\sigma$.

A symmetry of~$\mathcal{E}$ is an evolutionary vector field
\begin{equation}\label{eq:3}
  \Ev_\phi = \sum D_\sigma(\phi^j)\pd{}{u_\sigma^j}
\end{equation}
such that~$\ell_{\mathcal{E}}(\phi) = 0$, where $\phi=(\phi^1,\dots,\phi^m)$
is a function on~$\mathcal{E}$ which is called the generating section of the
symmetry at hand. Symmetries form an $\mathbb{R}$-Lie algebra with respect to
the commutator. This algebra is denoted by~$\sym(\mathcal{E})$. The commutator
of symmetries induces the Jacobi bracket of their generating sections denoted
by~$\{\cdot\,,\cdot\}$.

A horizontal $(n-1)$-form
\begin{equation*}
  \omega = \sum A_i\,dx^1\wedge\dots\wedge\,dx^i\wedge\dots\wedge\,dx^n
\end{equation*}
is a conservation law of~$\mathcal{E}$ if it is closed with respect to the
horizontal de~Rham differential
\begin{equation*}
  d_h = \sum dx^i\wedge D_i.
\end{equation*}
A conservation law is trivial if~$\omega$ is an exact form. Two conservation
laws are equivalent if their difference is a trivial conservation law.

Consider another equation~$\tilde{\mathcal{E}}$ and a locally trivial bundle
$\tau\colon \tilde{\mathcal{E}} \to \mathcal{E}$. It is called a
(differential) covering if $\tau_*(\tilde{D}_i) = D_i$ for any total
derivative on~$\tilde{\mathcal{E}}$. Two
coverings~$\tau_l\colon \tilde{\mathcal{E}}^l \to \mathcal{E}$, $l=1$, $2$,
are equivalent if there exists a diffeomorphism
$F\colon \tilde{\mathcal{E}}^1 \to \mathcal{E}$ such that (1)
$\tau_2\circ F = \tau_1$ and (2)
$F_*(\tilde{D}_i^1) = \sum \mu_i^j \tilde{D}_j^2$, where~$\mu_i^j$ are smooth
functions on~$\tilde{\mathcal{E}}^2$ and~$\tilde{D}_i^l$ are the total
derivatives on~$\tilde{\mathcal{E}}^l$. Symmetries of~$\tilde{\mathcal{E}}$
are said to be nonlocal symmetries of~$\mathcal{E}$ and similar for
conservation laws.

Denote by~$\mathcal{F}$ and~$\tilde{\mathcal{F}}$ the rings of smooth
functions on~$\mathcal{E}$ and~$\tilde{\mathcal{E}}$, respectively. Then an
$\mathbb{R}$-linear derivation~$S\colon \mathcal{F}\to \tilde{\mathcal{F}}$ is
a nonlocal shadow if
\begin{equation*}
  \tilde{D}_i\circ S = S\circ D_i,\qquad i= 1,\dots, n.
\end{equation*}
In particular, local symmetries can be regarded as shadows in any covering.
We say that a shadow~$S$ lifts to~$\tau$ if there exists a nonlocal
symmetry~$\tilde{S}$ such that $\eval{\tilde{S}}_{\mathcal{F}} = S$. Lifts of
the trivial shadow~$S=0$ are called invisible symmetries.

Denote by~$\{w^\beta\}$ coordinates in fibers of~$\tau$. They are called
nonlocal variables. Using these variables, we can write the
fields~$\tilde{D}_i$ as
\begin{equation*}
  \tilde{D}_i = D_i + \sum X_i^\beta\pd{}{w^\beta},
\end{equation*}
where~$X_i^\beta$ are smooth functions on~$\tilde{\mathcal{E}}$, while the
fact that~$\tau$ is a covering amounts to the compatibility of the system
\begin{equation*}
  w_{x^i}^\beta = X_i^\beta
\end{equation*}
modulo~$\tilde{\mathcal{E}}$. Then nonlocal $\tau$-symmetries are vector
fields
\begin{equation*}
  \tilde{\Ev}_\phi + \sum \psi^\beta\pd{}{w^\beta},
\end{equation*}
where~$\tilde{\Ev}_\phi$ is obtained from~\eqref{eq:3} by changing~$D_i$
to~$\tilde{D}_i$ and $\phi=(\phi^1,\dots,\phi^m)$, $\psi^\beta$ are functions
on~$\tilde{\mathcal{E}}$ that enjoy the system
\begin{align}
  \label{eq:4}
  \tilde{\ell}_{\mathcal{E}}(\phi)&=0,\\
  \label{eq:5}
  \tilde{D}_i(\psi^\beta)&=\tilde{\ell}_{X_i^\beta}(\phi) + \sum_\gamma
                           \pd{X_i^\beta}{w^\gamma}\psi^\psi,
\end{align}
where ``tilde'' denotes the natural lift of a differential operator in total
derivatives from~$\mathcal{E}$ to~$\tilde{\mathcal{E}}$. To describe shadows,
one must consider Equation~\eqref{eq:4} only, while invisible symmetries are
described by the equation
\begin{equation*}
  \tilde{D}_i(\psi^\beta) = \sum_\gamma\pd{X_i^\beta}{w^\gamma}\psi^\psi.
\end{equation*}

Let
\begin{equation*}
  \omega =(X_1\,dx^1 + X_2\,dx^2)\wedge\,dx^3\wedge\dots\wedge\,dx^n
\end{equation*}
be a two-component conservation law of~$\mathcal{E}$. Then one can construct
the covering~$\tau_\omega$ with the nonlocal variables~$w^\sigma$,
where~$\sigma$ is a symmetrical multi-index containing the integers
$3,\dots,m$, and the defining equations
\begin{equation*}
  w_{x^1}^\sigma = \tilde{D}_\sigma(X_1),\quad w_{x^2}^\sigma =
  \tilde{D}_\sigma(X_2),\quad w_{x^i}^\sigma = w^{\sigma i},
\end{equation*}
for $i\geq 3$. This is the Abelian covering associated with~$\omega$; it is
one-dimensional for~$n=2$ and infinite-dimensional otherwise.

\section{The universal hierarchy equation}
\label{sec:univ-hier-equat}

The universal hierarchy equation is of the form
\begin{equation}
  \label{eq:1}
  u_{yy} = u_t u_{xy} - u_y u_{tx},
\end{equation}
see~\cite{MartinezAlonsoShabat2002,MartinezAlonsoShabat2004}.

\subsection{Lax pair and the associated covering}
\label{sec:lax-pair-associated-UHE}

Equation~\eqref{eq:1} admits the following Lax pair
\begin{equation*}%\label{eq:6}
  \begin{array}{rcl}
  w_{t}&=&\lambda^{-2}(\lambda u_t-u_y)w_x,\\
  w_{y}&=&\lambda^{-1}u_yw_{x}.
  \end{array}
\end{equation*}
Expanding~$w$ in powers of~$\lambda$, $w = \sum_{i\in\mathbb{Z}}w_i\lambda^i$,
we obtain the infinite-dimensional covering
\begin{equation}\label{eq:8}
  \begin{array}{rcl}
    w_{i,t}&=&u_tw_{i+1,x}-u_yw_{i+2,x},\\
    w_{i,y}&=&u_yw_{i+1,x}
  \end{array}
\end{equation}
$i\in\mathbb{Z}$, with the additional variables~$w_i^{(j)}$ that satisfy the
relations $w_i^{(0)}=w_i$, $w_i^{(j+1)}=w_{i,x}^{(j)}$.

\subsection{Symmetries and reductions}
\label{sec:symm-reduct-UHE}

The space~$\sym(\mathcal{E})$ is spanned by the functions
$\theta_0(X)=Xu_x-X'u$, $\theta_1(X)=X$, $\phi_0(T)=Tu_t+T'yu_y$,
$\phi_1(T)=Tu_y$, $\upsilon=yu_y+u$, where~$X$ is a function in~$x$ and $T$ is
a function in~$t$, while `prime' denotes the corresponding derivatives.

\begin{lemma}
  \label{sec:lemma-symm-reduct-uhe}
  The symmetry~$\phi = \upsilon + \theta_0(1) + \phi_0(1)$ can be lifted to
  the covering~\eqref{eq:8}. 
\end{lemma}

\begin{proof}
  Denote the desired lift by
  \begin{equation*}
    \Phi = \tilde{\Ev}_\phi + \sum \phi^i\pd{}{w_i},
  \end{equation*}
  where~$\phi= yu_y + u + u_x + u_t$, and set
  \begin{equation*}
    \phi^i = (-i+1)w_i +yw_{i,y} + \frac{1}{u_y}w_{i-1,y} + w_{i,t}.
  \end{equation*}
  Then the result is obtained by the direct check.
\end{proof}

Due to Lemma~\ref{sec:lemma-symm-reduct-uhe}, we can consider the reduction of
Equation~\eqref{eq:1} together with its covering~\eqref{eq:8}. The resulting
objects will be the equation
\begin{equation}
  \label{eq:9}
  u_{yy} = u_y u_{xx} - (u_x + u)u_{xy} + u_x u_y 
\end{equation}
and the infinite-dimensional covering
\begin{equation}
  \label{eq:10}
  \begin{array}{rcl}
    q_{i,y}& = &(-i + 2)q_{i-1} - \dfrac{u_x + u}{u_y}q_{i-1,y} +
                 \dfrac{1}{u_y}q_{i-2,y} ,\\  
    q_{i,x} &=& \dfrac{q_{i-1,y}}{u_y}.
  \end{array}
\end{equation}
over this equation. Define the coverings~$\tau^p$ by setting~$q_i = 0$ for~$i<
p$, $p\in\mathbb{Z}$. Then, setting~$q_i^p=q_{p+i+1}$, we obtain~$q_{-1}=1$,
$q_0^p = -(p-1)y$ and
\begin{align*}
  \tau^p\colon\quad
  &q_{i,y}^p = (-p - i + 1)q_{i-1}^p - \frac{u_x + u}{u_y}
    q_{i-1,y}^p + \frac{1}{u_y}q_{i-2,y}^p,\\
  &q_{i,x}^p = \frac{q_{i-1,y}^p}{u_y}, 
\end{align*}
for $i\geq 1$. This is an infinite series of nonlocal conservation laws of
Equation~\eqref{eq:1}.

\begin{proposition}
  \label{sec:prop-symm-reduct-uhe}
  All the coverings~$\tau^p$ are pair-wise equivalent.
\end{proposition}

Before proving the result, consider an auxiliary construction. Namely,
introduce the operator
\begin{equation*}
  \mathcal{Y}^p = \sum_{i\geq 0}(i+1)q_{i+1}^p\pd{}{q_i^p}
\end{equation*}
and define the quantities $P_{i,j}$ as follows:
\begin{equation}\label{eq:7}
  P_{i,0}^p = \frac{1}{(i + 2)!}(q_0^p)^{i+2},\quad
  P_{i,j}^p = \frac{1}{j}\mathcal{Y}(P_{i,j-1}^p),\qquad
  i = 0,1,\dots,\quad j =  1,2,\dots
\end{equation}
We also assume $P_{i,j}^p=0$ when at least one of the subscripts is negative.

\begin{proof}[Proof of Proposition~\ref{sec:prop-symm-reduct-uhe}]
  We shall prove that any~$\tau^p$ is equivalent to~$\tau^0$. Two cases are to
  be considered.
  
  \emph{Case $p\neq 1$}. Let us set
  \begin{equation*}
    d_i = \sum_{l=0}^\infty (-1)^l\frac{(p + l)!}{p!}P_{l,i-l-1}^0.
  \end{equation*}
  Then
  \begin{equation*}
    q_i^p = -(p - 1)(q_i^0 - pd_i), i \geq 1,
  \end{equation*}
  is the desired equivalence.
  
  \emph{Case $p = 1$}. This way of proof does not work for $p=1$, but from the
  defining equations one can easily see that the covering~$\tau^1$
  \emph{coincides} with~$\tau^2$.
\end{proof}

\subsection{Weights}
\label{sec:weights-UHE}

Let us assign to all the local and nonlocal variables the weights
\begin{equation*}
  \abs{x} = 0,\quad \abs{y} = 1,\quad \abs{u} = -1,\quad \abs{q_i^p} = i+1.
\end{equation*}
We also set~$\abs{u_x} = \abs{u} - \abs{x}$, $\abs{u_y} = \abs{u} - \abs{y}$,
etc., and assume that the weight of a monomial is the sum of weights of its
factor. The weight of a vector field~$Z\pd{}{z}$ is $\abs{Z} - \abs{z}$. Then
all the constructions under consideration become graded with respect to these
weights, while the results (provided they are polynomial) split into
homogeneous components.

\subsection{Nonlocal symmetries of reductions}
\label{sec:nonl-symm-reduct-UHE}

Let us use the notation
\begin{equation}\label{eq:11}
  \Phi = (\phi,\phi^{p,1},\dots,\phi^{p,i},\dots)
\end{equation}
for the vector field
\begin{equation*}
  S = \tilde{\Ev}_{\phi} + \sum \phi^{p,i}\pd{}{q_i^p}
\end{equation*}
on~$\tau^p$. Then~\eqref{eq:11} is a symmetry if and only if
\begin{equation*}
  \tilde{D}_y^2(\phi)
  = u_y\tilde{D}_x^2(\phi) -(u_x+u)\tilde{D}_x\tilde{D}_y(\phi)
  +(u_y-u_{xy})\tilde{D}_x(\phi) + (u_x+u_{xx})\tilde{D}_y(\phi) -
  u_{xy}\phi
\end{equation*}
and
\begin{equation}
  \label{eq:12}
  \begin{array}{rcl}
    \tilde{D}_y(\phi^{p,1}) &
                              =
    & (p-1)\tilde{\mathcal{L}}_1(\phi),\\
    \tilde{D}_x(\phi^{p,1}) &
                              =
    &  -(p-1)\tilde{\mathcal{L}}_2(\phi);\\[4pt]
    \tilde{D}_y(\phi^{p,2}) &
                              =
    & -(p+1)\phi^{p,1} - \dfrac{u_x+u}{u_y}
      \tilde{D}_y(\phi^{p,1}) -q_{1,y}^p\tilde{\mathcal{L}}_1(\phi) -
      (p-1)\tilde{\mathcal{L}}_2(\phi),\\
    \tilde{D}_x(\phi^{p,2}) &
                            =
    & \dfrac{1}{u_y}\tilde{D}_y(\phi^{p,1}) +
      q_{1,y}^p\tilde{\mathcal{L}}_2(\phi);\\[4pt]
    \tilde{D}_y(\phi^{p,i}) &
                            =
    & (-p-i+1)\phi^{p,i-1} - \dfrac{u_x +
      u}{u_y}\tilde{D}_y(\phi^{p,i-1}) + \dfrac{1}{u_y}\tilde{D}_y(\phi^{p,i-2})\\
                            &&- q_{i-1,y}^p\tilde{\mathcal{L}}_1(\phi) +
                               q_{i-2}^p\tilde{\mathcal{L}}_2(\phi),\\
    \tilde{D}_x(\phi^{p,i}) &
                            =
    & \dfrac{1}{u_y}\tilde{D}_y(\phi^{p,i-1})
      +q_{i-1,y}^p\tilde{\mathcal{L}}_2(\phi), 
  \end{array}
\end{equation}
for all~$i>2$, where
\begin{equation*}
  \tilde{\mathcal{L}}_1 = \frac{1}{u_y} + \frac{1}{u_y}\tilde{D}_x - \frac{u_x
  +u}{u_y^2}\tilde{D}_y,\qquad \tilde{\mathcal{L}}_2 =
-\frac{1}{u_y^2}\tilde{D}_y
\end{equation*}
are the linearizations of the functions~$(u_x+u)/u_y$ and~$1/u_y$,
respectively, lifted to~$\tau^p$.

Direct computations show that the functions
\begin{equation*}
  \phi_{-1} = u_y,\quad \phi_0 = yu_y + u,\quad \psi_0 = u_x,\quad \psi_1 = e^{-x}
\end{equation*}
constitute a basis of the space~$\sym(\mathcal{E})$. In addition, it can be
checked that the function
\begin{align*}
  \phi_2^p &= (2p^2y^2 + py^2 - 3y^2 - 4q_1^p)(u_x + u) - 3py  + 3y\\
  & + \frac{1}{3}
  (5p^3y^3 + 6p^2y^3 - 8py^3 - 3y^3 - 15pyq_1^p - 27yq_1^p - 15q_2^p)u_y
\end{align*}
is a shadow in the covering~$\tau^p$. Here the subscripts indicate the weight
of the corresponding symmetry.

\begin{lemma}
  \label{sec:lem-nonl-symm-reduct-uhe-loc}
  The local symmetries~$\phi_0$\textup{,} $\psi_0$\textup{,} and~$\psi_1$ can
  be lifted to any covering~$\tau^p$.
\end{lemma}

\begin{proof}
  Let us set
  \begin{equation*}
    \begin{array}{lll}
      \phi_0^{p,i} = -(i+1)q_i^p + yq_{i,y}^p&\text{for}&\phi_0 = yu_y + u,\\
      \psi_0^p = q_{i,x}^p&\text{for}&\psi_0 = u_x,\\
      \psi_1^p = 0&\text{for}&\psi_1 = e^{-x}.
    \end{array}
  \end{equation*}
  Then it is an easy exercise to check that~\eqref{eq:12} fulfills. 
\end{proof}

\begin{lemma}
  \label{sec:lem-nonl-symm-reduct-uhe-nonloc}
  The symmetry~$\phi_{-1}$ can be lifted to the covering~$\tau^0$\textup{,}
  while the shadow~$\phi_2^3$ can be lifted to the covering~$\tau^3$.
\end{lemma}

\begin{proof}
  The lift of $\phi_{-1} = u_y$ is given by the formulas
  \begin{equation*}
    \phi_{-1}^{0,i} = q_{i,y}^0.
  \end{equation*}
  The lift of
  \begin{equation*}
    \phi_2^3 = (18y^2-4q_1^3)(u_x+u)+(54y^3-24yq_1^3-5q_2^3)u_y -6y
  \end{equation*}
  is given by
  \begin{align*}
    \phi_2^{3,i} &= 2(i+3)(2q_1^3-9y^2)q_i^3 -
                   6(i+4)yq_{i+1}^3 - 
                   2(i+5)q_{i+2}^3\\
                 &+ 
                   (54y^3-24yq_1^3-5q_2^3)q_{i,y}^3+2(9y^2-2q_1^3)q_{i,x}^3.
  \end{align*}
  Then~\eqref{eq:12} fulfills identically.
\end{proof}

\begin{lemma}
  \label{sec:nonl-symm-reduct-uhe-invis}
  The field $\Phi_{-2}^{-1} = \partial/\partial q_1^{-1}$ is an invisible
  symmetry in~$\tau^{-1}$.
\end{lemma}

\begin{proof}
  Direct check.  
\end{proof}

\begin{corollary}
  \label{sec:cor-nonl-symm-reduct-uhe}
  There exist symmetries~$\Phi_{-2}^p$\textup{,} $\Phi_{-1}^p$\textup{,}
  $\Phi_0^p$\textup{,} $\Phi_2^p$\textup{,} $\Psi_0^p$\textup{,}
  and~$\Psi_1^p$ in any covering~$\tau^p$.
\end{corollary}

\begin{proof}
  The fact follows immediately from Proposition~\ref{sec:prop-symm-reduct-uhe}
  and
  Lemmas~\ref{sec:lem-nonl-symm-reduct-uhe-loc}--\ref{sec:nonl-symm-reduct-uhe-invis}. 
\end{proof}

\begin{theorem}
  \label{sec:thm-nonl-symm-reduct-uhe}
  The Lie algebra of nonlocal symmetries for Equation~\eqref{eq:9} in~$\tau^p$
  is isomorphic to the direct sum
  \begin{equation*}
    \mathfrak{w} \oplus \mathfrak{s}_2,
  \end{equation*}
  where~$\mathfrak{w}$ is the Witt algebra and~$\mathfrak{s}_2$ is the
  two-dimensional solvable algebra.
\end{theorem}

\begin{proof}
  Since all the coverings~$\tau^p$ are pair-wise equivalent
  (Proposition~\ref{sec:prop-symm-reduct-uhe}), we can accomplish the proof in
  any of them. From the technical viewpoint,~$\tau^0$ is the most convenient
  one.
  
  Consider the transformation
  \begin{equation*}
    \tilde{\Phi}_0^0 = -\Phi_0^0-\Psi_0^0,\quad \tilde{\Phi}_{-1}^0 =
    -\Phi_{-1}^0,\quad \tilde{\Psi}_0^0 = \Psi_0^0,\quad \tilde{\Psi}_1^0 =
    \Psi_1^0. 
  \end{equation*}
  Let us
  set~$\tilde{\Phi}_1^0 = \frac{1}{3}\{\tilde{\Phi}_{-1}^0,\tilde{\Phi}_2^0\}$
  and by induction
  \begin{equation*}
    \tilde{\Phi}_{-k-1}^0 =-\frac{1}{k-1}\{\tilde{\Phi}_{-1}^0,
    \tilde{\Phi}_{-k}^0\},\qquad \tilde{\Phi}_{k+1}^0 
    = \frac{1}{k-1}\{\tilde{\Phi}_1^0,\tilde{\Phi}_k^0\}
  \end{equation*}
  for all~$k\geq 2$. Then
  \begin{equation*}
    \{\tilde{\Phi}_k^0,\tilde{\Phi}_l^0\} =(l-k)\tilde{\Phi}_{k+l}^0
  \end{equation*}
  for all~$k$, $l\in\mathbb{Z}$ and the functions~$\tilde{\Phi}_k^0$ span the
  algebra~$\mathfrak{w}$. On the other
  hand,~$\{\tilde{\Psi}_0^0,\tilde{\Psi}_1^0\} = \tilde{\Psi}_1^0$
  and~$\{\tilde{\Psi}_i^0,\tilde{\Phi}_k^0\} = 0$ for~$i=0$, $1$
  and~$k\in\mathbb{Z}$.
\end{proof}

\subsection{Explicit formulas}
\label{sec:explicit-formulas-uhe}

To conclude the discussion of the universal hierarchy equation, we present
explicit formulas for the lifts of symmetries~$\Phi_{-2}^p$, $\Phi_{-1}^p$,
$\Phi_1^p$, and~$\Phi_2^p$ to an arbitrary covering~$\tau^p$, $p\neq 1$:
\begin{align*}
  \phi_{-2}^p &= 0,\\
  \phi_{-2}^{p,1}&=1,\\
  \phi_{-2}^{p,i}
  &= \frac{p + 1}{p - 1} \left(q_{i-2} +\sum_{j=0}^{i-2} \left(\frac{-1}{p -
    1}\right)^{j+1} P_{j,i-j-3}^p \prod_{k=0}^j\big(-2 + k(p - 1)\big)\right);
\\[4pt]
  \phi_{-1}^p
  &= u_y,\\
  \phi_{-1}^{p,i}
  &=  q_{i,y}^p + pq_{i-1}^p + p\sum_{j=0}^{i-2}
    \left(\frac{-1}{p - 1}\right)^{j+1} P_{j,i-2-j}^p \prod_{k=0}^j\big(-1 +
    k(p - 1)\big);\\[4pt]
  \phi_1^p
  &=  (p-1)yu_x+\frac{1}{4}((p-1)(3p+2)y^2-6q^p_1)u_y+(p-1)\left(
    yu-\frac{1}{2} \right),\\ 
  \phi_1^{p,i}
  &= (p-1)yq^p_{i,x}+\frac{1}{4}\left((p-1)(3p+2)y^2-6q^p_{1}\right)q^p_{i,y}-\frac{1}{2}(3+i)(p-1)q^p_{i+1}\\
&+(p+i)q^p_0q^p_i -
  \frac{1}{2}(p-2)
  \sum_{j=0}^{i}{\left(\frac{-1}{p-1}\right)^j
  P_{j,i-j}^p\prod_{k=0}^{j}{[1+k(p-1)]}};\\[4pt]    
  \phi_2^p
  &=  \big((2p^2 + p - 3)y^2 - 4q_1^p\big)u_x + \frac{1}{3}\big((5p^3 +
    6p^2 - 8p - 3)y^3 - (15p - 27)yq_1^p - 15q_2^p\big)u_y \\
              &+ (2p^2y^2 + py^2 - 3y^2 - 4q_1^p)u - 3(p - 1)y,\\
  \phi_2^{p,i}
              &= \big((2p^2 + p - 3)y^2 - 4q_1^p\big)q_{i,x}^p + \frac{1}{3}
                \big((5p^3 + 6p^2 - 8p - 3)y^3 - 15(p - 27)yq_1^p -
                15q_2^p\big)q_{i,y}^p\\
              &- (p - 1)(2p + 3)(p + i)y^2q_i^p - 3(p - 1)(p + i + 1)yq_{i+1}^p
                + 4(p + i)q_1^pq_i^p - (5 + i)(p - 1)q_{i+2}^p \\
  &- (p - 3)\sum_{j=0}^{i+1}\left(\frac{-1}{p - 1}\right)^jP_{j,i-j+1}^p
    \prod_{k=0}^j\big(2 + k(p - 1)\big),
\end{align*}
where the quantities~$P_{i,j}^p$ are described by Equations~\eqref{eq:7}.

\section{The rdDym equation}
\label{sec:rddym-equation}

The 3D rdDym equation reads
\begin{equation}
  \label{eq:2}
  u_{ty} = u_x u_{xy} - u_y u_{xx},
\end{equation}
see~\cite{Blaszak,Ovsienko2010,Pavlov2006}.

\subsection{Lax pairs and associated coverings}
\label{sec:lax-pairs-associated-rd}

The following system
\begin{equation*}
  \begin{array}{rcl}
    w_t &=& (u_x - \lambda)w_x,\\
    w_y &=& \lambda^{-1} u_y w_x
  \end{array}
\end{equation*}
is a Lax pair for Equation~\eqref{eq:2}. As above, we consider the expansion
$w= \sum_{i\in\mathbb{Z}} w_i\lambda^i$ and obtain the covering
\begin{equation}\label{eq:13}
  \begin{array}{rcl}
    w_{i,t} &=& u_x w_{i,x} - w_{i-1,x},\\
    w_{i,y} &=& u_y w_{i+1,x}
  \end{array}
\end{equation}
$i\in\mathbb{Z}$, endowed with the additional nonlocal variables~$w_i^{(j)}$
defines by the relations~$w_i^{(0)} = w_i$, $w_i^{(j+1)} = w_{i,x}^{(j)}$.

\subsection{Symmetries and reductions}
\label{sec:symm-reduct-rd}

The space~$\sym(\mathcal{E})$ for Equation~\eqref{eq:2} is spanned by the
functions $\psi_0 = xu_x - 2u$, $\upsilon_0(Y) = Yu_y$,
$\theta_0(T) = Tu_t + T'(xu_x - u) + \frac{1}{2}T''x^2$,
$\theta_{-1}(T) = Tu_x + T'x$, $\theta_{-2}(T) = T$, where~$T = T(t)$,
$Y = Y(y)$, and the `prime' denotes the derivative with respect to~$t$.

\begin{lemma}
  \label{sec:lem-symm-reduct-rd}
  The symmetry~$\phi = \theta_0(1) - \upsilon_0(1) + \psi_0$ can be lifted to
  the covering~\eqref{eq:13}.
\end{lemma}

\begin{proof}
  Let~$\phi^i$ denote the coefficient at~$\pd{}{w_i}$. Then
  \begin{equation*}
    \phi^i = w_{i,t} -w_{i,y} -xw_i^{(1)} -(i+2)w_i
  \end{equation*}
  delivers the desired lift.
\end{proof}

The reduction with respect to the obtained lift leads to the equation
\begin{equation}
  \label{eq:14}
  u_{yy} = (u_x + x)u_{xy} - (u_{xx} + 2)u_y 
\end{equation}
and the covering
\begin{equation}
  \label{eq:15}
  \begin{array}{rcl}
    r_{i,x} &=& (u_x + x)r_{i-1,x} - r_{i-1,y} - (i + 1)r_{i-1},\\
    r_{i,y} &=&u_y r_{i-1,x}
  \end{array}
\end{equation}
over~\eqref{eq:14}. Similar to Subsection~\ref{sec:symm-reduct-UHE}, we fix an
integer~$p$ and `cut' this covering at level~$p$, i.e., set~$r_i = 0$ for
all~$i<p$. Then, after relabeling~$r_{p+i}\mapsto r_{i-2}^p$, we obtain that
\begin{equation*}
  r_{-2}^p = 1,\quad r_{-1}^p = -(p + 2)x,\quad r_0^p = -(p + 2)u +
  \frac{1}{2}(p + 2)^2x^2
\end{equation*}
and arrive to the coverings
\begin{align*}
  \rho^p\colon\quad& r_{i,x}^p = (u_x + x)r_{i-1,x}^p - r_{i-1,y}^p - (p + i +
                3)r_{i-1}^p,\\ 
              &r_{i,y}^p = u_y r_{i-1,x}^p,
\end{align*}
$i\geq1$. These are nonlocal conservation laws of Equation~\eqref{eq:14}.

\begin{proposition}
  \label{sec:prop-symm-reduct-rd}
  All the coverings~$\rho^p$ are pair-wise equivalent.
\end{proposition}

\begin{proof}
  The proof is very similar to that of
  Proposition~\ref{sec:prop-symm-reduct-uhe}. Two cases must be considered.

  \emph{Case~$p\neq-2$}. Consider the vector field
  \begin{equation*}
    \mathcal{Z}^p = \sum_{i=-1}^\infty(i+2)r_{i+1}^p\pd{}{r_i^p}
  \end{equation*}
  and define the quantities~$Q_{i,j}$ by
  \begin{equation}
    \label{eq:16}
    Q_{i,0}^p = \frac{1}{(i + 2)!}(r_{-1}^p)^{(i+2)},\quad
    Q_{i,j}^p =\frac{1}{j}\mathcal{Z}(Q_{i,j-1}^p),\qquad i = 0,1,\dots,\quad j =
    1,2,\dots
  \end{equation}
  and formally set~$Q_{i,j}^p=0$ when at least one of the subscripts is
  negative.  Let
  \begin{equation*}
    d_i = \sum_{k=0}^\infty (-1)^i\frac{(p + k + 3)!}{(p + 3)!} Q_{i,i-k}^{-3}.
  \end{equation*}
  Then
  \begin{equation*}
    r_i^p = -(p + 2)\big(r_i^{-3} - (p + 3)d_i\big)
  \end{equation*}
  is an equivalence between~$\rho^p$ and~$\rho^{-3}$.

  \emph{Case~$p=-2$}. It is easily seen that~$\rho^{-2}$ coincides
  with~$\rho^{-1}$.
\end{proof}

\subsection{Weights}
\label{sec:weights-rd}

The basic weights assigned in this case are
\begin{equation*}
  \abs{x} = 1,\quad \abs{y} =0,\quad \abs{u} = 2
\end{equation*}
with the same rules that were described in Subsection~\ref{sec:weights-UHE}.

\subsection{Nonlocal symmetries of reductions}
\label{sec:nonl-symm-reduct-rd}

Note first that
\begin{equation*}
  \Phi = (\phi,\phi^{p,1},\dots,\phi^{p,i},\dots)
\end{equation*}
is a symmetry in~$\rho^p$ if and only if
\begin{equation*}
  \tilde{D}_y^2(\phi) = (u_x + x)\tilde{D}_x\tilde{D}_y(\phi) -
  u_y\tilde{D}_x^2(\phi) + u_{xy}\tilde{D}_x(\phi) - (u_{xx} + 2)\tilde{D}_y(\phi) 
\end{equation*}
and
\begin{equation}
  \label{eq:17}
  \begin{array}{rcl}
    \tilde{D}_x(\phi^{p,1})&=&(p+2)\big(((p+2)x - 2u_x)\tilde{D}_x(\phi) +
                               \tilde{D}_y(\phi) (p+4)\phi\big),\\
    \tilde{D}_y(\phi^{p,1})&=& (p+2)\big(((p+2)x - u_x)\tilde{D}_y(\phi) -
                               u_y\tilde{D}_x(\phi)\big);\\[4pt] 
    \tilde{D}_x(\phi^{p,i})&=& (u_x + x)\tilde{D}_x(\phi^{p,i-1}) -
                               \tilde{D}_y(\phi^{p,i1})  - (p + i +
                               3)\phi^{p,i-1} + r_{i-1,x}^p\tilde{D}_x(\phi),\\
    \tilde{D}_y(\phi^{p,i})&=&u_y\tilde{D}_x(\phi^{p,i-1}) +
                               r_{i-1,x}^p\tilde{D}_y(\phi), 
  \end{array}
\end{equation}
where~$i>1$.

A basis of $\sym(\mathcal{E})$ is formed by the functions
\begin{equation*}
  \phi_{-2} = 1,\quad \phi_{-1} =u_x +x,\quad \phi_0 = 2u-xu_x,\quad \psi_0=u_y,
\end{equation*}
where subscripts coincide with weights. In addition, in any covering~$\rho^p$
there exists a shadow of the form
\begin{align*}
  \phi_2^p
  &= 6r_2^p + \left(
    (p + 2)
    \left(
    (5p + 24)xu - \frac{1}{6} (5p^2 + 20p - 18)x^3
    \right)
    - 5r_1^p
    \right)u_x\\
  &- (p + 2)(4u + 5x^2)u_y + (6p + 25)xr_1^p\\
    &+ (p + 2)\left(
    (3p + 16)u^2 - (p + 4)(3p + 13)x^2u +
    \frac{1}{12} (9p^3 + 80p^2 + 212p + 168)x^4
    \right).
\end{align*}

\begin{lemma}
  \label{sec:nonl-symm-reduct-rd-loc}
  The symmetries~$\psi_0$ and~$\phi_0$ are lifted to any covering~$\rho^p$.
\end{lemma}

\begin{proof}
  It is sufficient to set
  \begin{equation*}
    \begin{array}{lll}
      \psi_0^{p,i}=r_{i,y}^p&\text{for}&\psi_0=u_y,\\
      \phi_0^{p,i}=-xr_{i,x}^p + (i+2)r_i^p&\text{for}&\phi_0=2u-xu_x
    \end{array}
  \end{equation*}
  and check that Equations~\eqref{eq:17} fulfill.
\end{proof}

\begin{lemma}
  \label{sec:nonl-symm-reduct-rd-nloc}
  The symmetries~$\phi_{-1}$ and~$\phi_{-2}$ are lifted to the
  coverings~$\rho^{-3}$ and~$\rho^{-4}$\textup{,} respectively\textup{,} while
  the shadow~$\phi_2^0$ lifts to~$\rho^0$.
\end{lemma}

\begin{proof}
  We set
  \begin{equation*}
    \begin{array}{lll}
      \phi_{-1}^{-3,i} = r_{i,x}^{-3}
      &\text{for}
      &\phi_{-1}=u_x+x,\\
      \phi_{-2}^{-4,i} = 0
      &\text{for}
      &\phi_{-2}=1
    \end{array}
  \end{equation*}
  and also
  \begin{align*}
    \phi_2^{0,i} &= (48xu + 6x^3 - 5r_1^0)r_{i, x}^0
                   - 2(4u + 5x^2)r_{i, y}^0 +
                   2(i+6)r_{i+2}^0 - 2(36+9i)x^2r_i^0 \\
                 &+ 3(5+i)r_{-1}^0r_{i+1}^0 +
                   4(i+4)r_0^0r_i^0
  \end{align*}
  for
  \begin{equation*}
        \phi_2^0 = 6 r_2^0+(48 x u+6 x^3-5 r_1^0)u_x
        - 2(4u+5
        x^2)u_y + 25xr_1^0+32 u^2
        -104 x^2 u+28 x^4.
  \end{equation*}
  Then Equations~\eqref{eq:17} are satisfied for the corresponding values
  of~$p$.
\end{proof}

\begin{corollary}
  All the symmetries~$\Psi_0^p$\textup{,} $\Phi_{-2}^p$\textup{,}
  $\Phi_{-1}^p$\textup{,} $\Phi_0^p$\textup{,} $\Phi_2^2$ exist in
  any~$\rho^p$.
\end{corollary}

\begin{proof}
  It immediately follows from Proposition~\ref{sec:prop-symm-reduct-rd} and
  Lemmas~\ref{sec:nonl-symm-reduct-rd-loc}
  and~\ref{sec:nonl-symm-reduct-rd-nloc}.
\end{proof}

We can describe the algebra of nonlocal symmetries for Equation~\eqref{eq:14}
now:

\begin{theorem}\label{sec:thm-nonl-symm-reduct-rd}
  The algebra of nonlocal symmetries of Equation~\eqref{eq:14} in any
  covering~$\rho^p$ is isomorphic to
  \begin{equation*}
    \mathfrak{w} \oplus \mathfrak{a}_1,
  \end{equation*}
  where~$\mathfrak{w}$ is the Witt algebra and~$\mathfrak{a}_1$ is the
  one-dimensional Abelian Lie algebra.
\end{theorem}

\begin{proof}
  Similar to the proof of Theorem~\ref{sec:thm-nonl-symm-reduct-uhe}, we first
  choose a convenient value of~$p$, which is~$p=-1$ in our case, and
  set~$\tilde{\Phi}_2^{-1} = -\Phi_2^{-1}$.  The $\mathfrak{w}$-component is
  constructed exactly in the same way as it was done in the proof of
  Theorem~\ref{sec:thm-nonl-symm-reduct-uhe}. The Abelian component is spanned
  by the symmetry~$\Psi_0^{-1}$ which obviously commutes with
  all~$\Phi_i^{-1}$.
\end{proof}

\subsection{Explicit formulas}
\label{sec:explicit-formulas-rd}

Let us describe the lifts~$\Phi_i^p$, $i=-2$, $-1$, $1$, $2$, explicitly:
\begin{align*}
  \phi_{-2}^{p,i}
  &=  (p + 4)\left(-r_{i-2}^p - \sum_{j=1}^\infty\left(\left(\frac{-1}{p +
    2}\right)^j Q_{j-1,i-j-1}^p \prod_{l=0}^{j-1}
    \big(l(p + 2) - 2\big)\right)\right),\\
  \phi_{-1}^{p,i}
  &=r_{i,x}^p + (p + 3)\left(r_{i-1}^p + \sum_{j=1}^\infty\left(
    \left(\frac{-1}{p + 2}\right)^j Q_{j-1,i-j}^p \prod_{l=0}^{j-1}
    \big(l(p + 2) - 1\big)\right)\right),\\
  \phi_1^{p,i}
  &= (p + 2)\left(\left(3u +\frac{5}{2}x^2\right)r_{i,x}^p - 2xr_{i,y}^p - (i
    + 4)r_{i+1}^p\right) + 2(p + i + 4)r_{-1}^pr_i^p\\
    &- (p + 1)\left(\sum_{j=0}^\infty\left(\frac{-1}{p + 2}\right)^j
      Q_{j,i-j+1}^p\prod_{l=0}^j\big(l(p + 2) + 1\big)\right),\\
  \phi_2^{p,i}
  &=\left((p + 2)\left((5p + 24)xu - \frac{1}{6}(5p^2 + 20p - 18)x^3\right)
    - 5r_1^p\right)r_{i,x}^p \\
  &- (p + 2)\left((4u + 5x^2)r_{i,y}^p - (i + 6)r_{i+2}^p\right) 
    - (p + 2)\big((2ip + 13p + 9i + 36)x^2+ 4pu\big)r_i^p\\
  &+ 3(p + i + 5)r_{-1}^pr_{i+1}^p + 4(i + 4)r_0^p r_i^p
    - p\sum_{j=0}^\infty\left(\left(\frac{-1}{p + 2}\right)^jQ_{j,i-j+2}^p
    \prod_{l=0}^j\big(l(p + 2) + 2\big)\right).
\end{align*}
Here the quantities~$Q_{i,j}^p$ are given by Equations~\eqref{eq:16}.

\section{Discussion}
\label{sec:discussion}

Let us conclude with several remarks:
\begin{itemize}
\item All the nonlocal symmetry algebras of linear degenerate equations
  (see~\cite{B-K-M-V-2018}) and their reductions (see~\cite{H-K-M-V-jnmp} and
  the results above) contain the Witt algebra~$\mathfrak{w}$ as their
  semi-direct (or direct) summand.
\item In all the constructions used to describe the symmetry algebras
  structures the crucial role is played by the operators similar
  to~$\mathcal{Y}$ and~$\mathcal{Z}$ from Sections~\ref{sec:symm-reduct-UHE}
  and~\ref{sec:symm-reduct-rd} and the quantities~$P_{i,j}^p$ and~$Q_{i,j}^p$. It
  is interesting to understand the geometric origins of these objects.
\item It is also interesting to study other Lax integrable equations in
  dimension~$>2$ that are not linear degenerate and compare their nonlocal
  symmetry structure with the already known results.
\end{itemize}
We plan to shed the light on the last two items in the forthcoming research.

\section*{Acknowledgments}

Computations were supported by the \textsc{Jets} software,~\cite{Jets}.

The first author (PH) was supported by the Specific Research grant SGS/6/2017
of the Silesian University in Opava.  The second author (ISK) was partially
supported by the \emph{2017 Dobrushin professor} grant. The third author (OIM)
is grateful to the Polish Ministry of Science and Higher Education for
financial support.

\end{document}